\documentclass[conference]{IEEEtran}
\makeatletter
\def\ps@headings{%
\def\@oddhead{\mbox{}\scriptsize\rightmark \hfil \thepage}%
\def\@evenhead{\scriptsize\thepage \hfil \leftmark\mbox{}}%
\def\@oddfoot{}%
\def\@evenfoot{}}
\makeatother 

\usepackage{times,amsmath,amssymb,graphicx,epsf,psfrag,epsfig,cite,color}
\ifCLASSINFOpdf

\else

\fi

\def\boxit#1{\vbox{\hrule\hbox{\vrule\kern3pt
        \vbox{\kern3pt#1\kern3pt}\kern3pt\vrule}\hrule}}

\def\reals{ { {\rm  I \kern-0.15em R }  } }
\def\complex{ {\,{{\rm C} \kern-0.50em \raise0.20ex {  |}}\, }}

\def\Rbf{{\bf R}}

\def\Ac{{\cal A}}

\def\Pc{{\cal P}}

\def\be{\begin{equation}}
\def\ee{\end{equation}}

\def\defeq{{\stackrel{\Delta}{=}}}
\def\scalefig#1{\epsfxsize #1\textwidth}
\def\Rxx{\Rbf_{\ssstyle X\kern-.1em X}}

\let\ssstyle=\scriptscriptstyle

\def\eg{{\it e.g.,\ \/}}

\def\ie{{\it i.e.,\ \/}}
\def\Kout{\setbox1=\hbox{\Huge\bf K}\hbox to
1.05\wd1{\hspace{.05\wd1}
}}
\def\Sout{\setbox1=\hbox{\Huge\bf S}\hbox to 1.05\wd1{\hspace{.05\wd1}
}}

\def\ie{{\it i.e.,\ \/}}
\def\etc{{\it etc.}}
\def\defeq{{\,\stackrel{\Delta}{=}}\,}

\def\scalefig#1{\epsfxsize #1\textwidth}
\def\nn{{\nonumber}}

\newtheorem{lemma}{Lemma}
\newtheorem{theorem}{Theorem}

\newtheorem{definition}{Definition}

\begin{document}

\title{Dynamic Intrusion Detection in Resource-Constrained Cyber Networks}

\author{\IEEEauthorblockN{Keqin Liu,~Qing Zhao}
\IEEEauthorblockA{Electrical and Computer Engineering, University of California, Davis, CA 95616\\
Email: \{kqliu,qzhao\}@ucdavis.edu } }

\maketitle

\begin{abstract}
We consider a large-scale cyber network with $N$ components (\eg
paths, servers, subnets). Each component is either in a healthy
state ($0$) or an abnormal state ($1$). Due to random intrusions,
the state of each component transits from $0$ to $1$ over time
according to certain stochastic process. At each time, a subset of
$K~(K<N)$ components are checked and those observed in abnormal
states are fixed. The objective is to design the optimal scheduling
for intrusion detection such that the long-term network cost
incurred by all abnormal components is minimized. We formulate the
problem as a special class of Restless Multi-Armed Bandit (RMAB)
process. A general RMAB suffers from the curse of dimensionality
(PSPACE-hard) and numerical methods are often inapplicable. We show
that, for this class of RMAB, Whittle index exists and can be
obtained in closed form, leading to a low-complexity implementation
of Whittle index policy with a strong performance. For homogeneous
components, Whittle index policy is shown to have a simple structure
that does not require any prior knowledge on the intrusion
processes. Based on this structure, Whittle index policy is further
shown to be optimal over a finite time horizon with an arbitrary
length. Beyond intrusion detection, these results also find
applications in queuing networks with finite-size buffers.
\end{abstract}

\section{Introduction}
The objective of Intrusion Detection Systems (IDS) is to locate
malicious activities (\eg denial of service attack, port scans,
hackers) in the quickest way such that the infected parts can be
timely fixed to minimize the overall damage to the network. With the
increasing size, diversity, and interconnectivity of the cyber
system, however, intrusion detection faces the challenge of
scalability: how to rapidly locate intrusions and anomalies in a
large dynamic network with limited resources. The two basic
approaches to intrusion detection, namely, active probing and
passive monitoring~\cite{JajodiaEtal09Book,{DebarEtal05CN}}, face
stringent resource constraints when the network is large and
dynamic. Specifically, active-probing based approaches need to
choose judiciously which components of the network to probe to
reduce overhead; passive-monitoring based approaches need to
determine how to sample the network so that real-time processing of
the resulting data is within the computational capacity of the
IDS~\cite{KodialamLaksham03INFOCOM}. The problem is compounded by
the fact that the adversarial behaviors are typically random and
evolving.

In this paper, we address resource-constrained intrusion detection
in large dynamic cyber networks. Specifically, we consider a network
with $N$ heterogeneous components which can be paths, routers, or
subnets. At a given time, a component can be in a healthy state or
an abnormal state. An abnormal component remains abnormal until the
anomaly is detected and resolved. A healthy component may be
attacked and become abnormal if the attack is successful. We
consider a general attack model: the behavior of the intruder can be
arbitrarily correlated in time and varies across components, and
different attacks can be launched with different probabilities of
successfully compromising the component under attack. As a
consequence, the state of a component evolves according to an
arbitrary stochastic process until it is probed/sampled. When a
healthy component is probed/sampled, its state evolution (\ie how
likely it will become abnormal in each subsequent time instant) is
reset. This models the scenario where proactive actions are taken
(patches are installed, firewalls upgraded, \etc) by the IDS when
probing/sampling a component to refresh its immunity to attacks.
Note that this model is significantly different and more complicated
than the SIS (susceptible-infected-susceptible) model and its
variants (see, \eg~\cite{Wierman&Marchette}).

For each component in an abnormal state, a cost (depending on the
criticality of the component) per unit time is incurred. At each
time, the IDS can choose a subset of $K$ components to probe or
sample ($K$ is often much smaller than $N$ due to resource
constraints). The question here is how to dynamically probe or
sample these $N$ components to minimize the long-term cost over
time. The key is to learn from past observations and decisions and
dynamically adjust the probing/sampling actions.

\subsection{Main Results}
We formulate the dynamic intrusion detection problem as a special
class of Restless Multi-Armed Bandit (RMAB) process, where each
component is considered as an arm. While finding the optimal
solution to a general RMAB problem is PSPACE-hard with {\em
exponential} complexity in system size~\cite{Complexity}, we show
that for this class of RMAB at hand, several structural properties
exist that lead to simple robust solutions. Specifically, by
exploring the reset nature of the problem, we first show that a
sufficient statistic for choosing the optimal probing/sampling
actions is given by a two-dimensional vector of each arm that can be
easily updated at each time. This significantly reduces the state
space for optimal decision making. Second, we show that this RMAB is
indexable, thus an index policy---referred to as Whittle index
policy---with strong performance and {\em linear} complexity in the
size $N$ of the cyber network can be constructed. Third, we show
that the Whittle index can be obtained in closed form, leading to
negligible complexity of implementation. Fourth, we show that for
homogeneous components, the low-complexity Whittle index policy has
a simple robust structure that does not need any prior knowledge on
the stochastic attack model and achieves the optimal performance.

In the context of RMAB, our results contribute to the study of the
existence and optimality of Whittle index policy. In 1988, Whittle
generalized the classic MAB to RMAB, a more powerful stochastic
model to take into account system dynamics that cannot be directly
controlled~\cite{whittle}. Whittle proposed an index policy that has
been shown to be asymptotically (when the system size approaches
infinity) optimal under certain
conditions~\cite{Weber,{Weber:addendum}}. The difficulty of Whittle
index policy lies in the complexity of establishing its existence
(the so-called indexability) and computing the index. There is no
general characterization regarding which class of RMAB is indexable,
and little is known about the optimality of Whittle index (when it
does exist) for finite-size systems. In this paper, we present a
significant class of indexable RMAB with practical applications for
which Whittle index policy is shown to be optimal for homogeneous
arms. This result lends a strong justification for the existence and
the optimality of linear complexity algorithms based on the Whittle
index. Beyond intrusion detection, this special class of RMAB and
the corresponding results can also be applied to the holding cost
minimization problem in queuing networks with finite-size buffers,
as elaborated in Sec.~\ref{sec:queue}.

\subsection{Related Work}\label{sec:relate}
In~\cite{Lee&Stolfo98USS}, the problem of intrusion recognition by
classifying system patterns was addressed based on data mining.
Without resource constraint, the focus is on the best selection of
system features to detect intrusion from the accessible system data
statistics. Similar problems of statistical modeling of data and
detection algorithms under various scenarios were considered in a
number of papers, \eg
\cite{Denning87TSE,{Ghosh&Etal99IDNM},{Roesch99},{Bass00COMACM}}.
These studies mainly address the intrusion detection problem from a
machine learning or pattern recognition perspective and do not
consider the constraint on the system monitoring capacity. Our work
is a stochastic control approach for intrusion detection in large
networks with resource constraints, where the problem of how to
adaptively allocate the limited detecting and repair power for
performance optimization is of great interest.
In~\cite{BarfordEtal09}, a set of heuristic detection, path
selection and link anomaly localization algorithms were proposed
based on the active probe-enabled network measurements.
In~\cite{AlpcanBasar06}, the intrusion detection problem was
formulated as a zero-sum game with two players (the intruder and the
IDS), where the game evolutions and outcomes were studied through
numerical examples based on Markovian decision processes and
Q-learning. The previous algorithm designs mainly take into account
the static or Markovian dynamics of the networks. The results in
this paper thus represent a step forward over the previous work by
addressing the general non-Markovian network dynamics.

In the literature of RMAB, the indexability was studied
in~\cite{Nino}, where efficient algorithms were constructed to
numerically test indexability and compute Whittle index for
finite-state systems. For the problem at hand, the system state
space is infinite, and thus numerical methods are generally
infeasible, even for a fixed realization of system parameters. We
show that, however, indexability holds regardless of the system
parameters and Whittle index can be solved in closed-form. The
optimality of Whittle index policy was subsequently established for
homogenous arms. For a special class of RMAB as detailed in the next
paragraph, the optimality of Whittle index policy was established
for homogeneous arms under certain conditions. In general, the
optimality of Whittle index policy has rarely been established.
Nevertheless, numerical studies have demonstrated the
near-optimality of Whittle index policy for numerous RMAB models
(see,
\eg~\cite{HeEtal11INFOCOM,{GlazebrookMitchell:02},Ansell,{Glazebrook1}}).

In the context of dynamic spectrum access and multi-agent tracking
systems, a class of RMAB modeled by a two-state Markovian model was
considered in~\cite{Liu&Zhao:08IT,{Leny08ACC}}. The indexability was
established and Whittle index was solved in closed form. The
Markovian model yields special structural properties of the system
dynamic equations that significantly simplify the establishment of
the indexability and Whittle index. However, these structural
properties no longer hold for the RMAB considered here that deals
with arbitrary underlying random processes, and the approaches
in~\cite{Liu&Zhao:08IT,{Leny08ACC}} do not apply. In this paper, we
propose a new approach for establishing the indexability and the
closed-form Whittle index based on a comparing argument on the
optimal stopping times. Besides the RMAB model at hand, this
approach is extendable to general two-state reset processes with
partially observable states. In~\cite{Liu&Zhao:08IT}, Whittle index
policy was shown to be equivalent to the myopic policy for
homogeneous arms, which leads to its optimality under certain
conditions based on the previous results on the myopic policy
established
in~\cite{Zhao&etal:08TWC,{Ahmad&etal},{Ahmad&Liu:09Allerton}}.
Again, the approaches
in~\cite{Zhao&etal:08TWC,{Ahmad&etal},{Ahmad&Liu:09Allerton}} are
based on the special properties, \eg the linearity of the value
function, of the myopic policy under the Markovian model. For the
problem at hand, although the equivalence between Whittle index
policy and the myopic policy is preserved for homogeneous arms, the
properties under the Markovian model no longer hold. To show the
optimality, we take a different approach by establishing the
monotonicity of the value function, as detailed in
Sec.~\ref{sec:optwhittle}.

\section{Network Model}\label{sec:model}

Consider a cyber network with $N$ inhomogeneous components that are
subject to random attacks over time. At each discrete time, each
component is either in the healthy state ($0$) or the abnormal state
($1$). If an attack to a healthy component is successful, the
component enters the abnormal state until it is probed and fixed. We
assume that different components experience statistically
independent but not necessarily identical attack processes.

Each attack process can be arbitrarily correlated over time.
Consequently, the state evolution of a component is given by an
arbitrary probability sequence $\{p_n(t)\}_{t\ge0}$, where $p_n(t)$
is the probability that component $n$ enters state $1$ after $t$
steps since the last time it was probed. Specifically, if a
component (say, component $n$) is probed and observed in state~$0$,
a simple maintenance action is taken which resets its state
evolution according to $\{p_n(t)\}_{t\ge0}$. If component $n$ is
observed in state~$1$, a sophisticated repair action is taken, and
the component will be back to the normal state in the next time
instant\footnote{Parallel results can be obtained for the model in
which a repaired component cannot be guaranteed to be healthy in the
next time instant and are omitted here due to the space limit.} and
then evolve according to $\{p_n(t)\}_{t\ge0}$. Note that
$\{p_n(t)\}_{t\ge0}$ is a monotonically increasing sequence since
state~$1$ is absorbing when the component is unobserved. A simple
example is given by the i.i.d. attack process, where component $n$
is compromised with a constant probability $q_n\in(0,1)$ at each
time. For this example, the state of component $n$ transits as a
Markov chain shown in Fig.~\ref{fig:Markov}, and we have
\begin{eqnarray}\nn
p_n(t)&=&1-(1-q_n)^t,
\end{eqnarray}which monotonically converges to $1$ at the
geometric rate ($1-q_n$) as $t$ increases. In general, we do not
require any specific form of $\{p_n(t)\}_{t\ge0}$.

\begin{figure}
\centerline{
\begin{psfrags}
\psfrag{A}[c]{$0$}\psfrag{B}[c]{$1$}\psfrag{q}[c]{$q_n$}\psfrag{1-q}[c]{$1-q_n$~~}
\psfrag{1}[c]{$1$} \scalefig{.4}\epsfbox{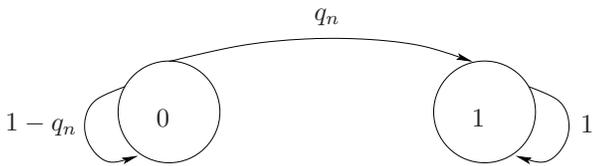}
\end{psfrags}}
\caption{An example based on the Markovian state model.}
\label{fig:Markov}
\end{figure}

For each abnormal component (say, component $n$), a cost $c_n$ is
incurred per unit time. With limited resource, only a subset of
$K~(K<N)$ components can be probed for maintenance/repair. The
objective is to minimize the long-term average network cost by
designing the optimal sequential component probing policy.

\section{RMAB Formulation}\label{sec:formulation}
In this section, we formulate the intrusion detection problem as a
special class of Restless Multi-Armed Bandit (RMAB) process. The
concepts of indexability and Whittle index are also introduced.

\subsection{RMAB and Sufficient Statistics}
In a general RMAB, a player chooses $K$ out of $N$ independent arms
to activate at each time based on the current states of all arms. At
each time, the state of each arm transits according to two
potentially different Markovian rules depending on whether it is
made active or passive. Each arm contributes an immediate reward
depending on its current state and the imposed action. The objective
is to maximize the long-term reward by optimally selecting arms to
activate over time based on the arm state evolutions.

We need to note that the states of all arms are assumed to be
completely observable and obey Markovian transition rules in an
RMAB. However, for the intrusion detection problem at hand, the
state ($0/1$) of each component is not observable unless it is
probed, and the state transition rules are non-Markovian in general.
It is thus not suitable to model the component state as the arm
state. By exploring the reset nature of the problem, we show in the
next lemma that a sufficient statistic for optimal decision making
is given by the two-dimensional vector set $\{(i_n,t_n)\}_{n=1}^N$,
where $i_n\in\{0,1\}$ is the last observed state of component $n$
and $t_n$ the time lapsed since the last observation. As a
consequence, we can treat $(i_n,t_n)$ as the arm state of component
$n$, which is complete observable but with an infinite dimension. In
the rest of paper, we refer to $(i_n,t_n)$ as the {\em arm state} of
component $n$ to distinguish it from the component state
$S_n\in\{0/1\}$. We also let
$a_n\in\{\mbox{active/probe}~(1),\mbox{passive/not probe}~(0)\}$
denote the probing action on arm $n$.
\begin{lemma}
For the intrusion detection problem, the vector set
$\{(i_n,t_n)\}_{n=1}^N$ is a sufficient statistics for optimal
decision making. Furthermore, given the current probing actions and
observations, the arm state $(i_n,t_n)$ of component $n$ transits
according to the following Markovian rules.
\begin{eqnarray}\nn
\Gamma(i_n,t_n)=\left\{\begin{array}{ll}(0,1),~&~\mbox{if}~a_n=1,~S_n=0
\\(1,1),~&~\mbox{if}~a_n=1,~S_n=1\\
(i_n,t_n+1),~&~\mbox{if}~a_n=0\end{array}\right.,
\end{eqnarray}where $\Gamma(\cdot)$ denotes the one-step transition
of the arm state given the current arm state and action.
\end{lemma}
\begin{proof}
Recall that each active action on each component (say, component
$n$) resets its state evolution according to the probability
sequence $\{p_n(t)\}_{t\ge0}$ (see Sec.~\ref{sec:model}). Given
$(i_n,t_n)$, the future state statistics of component $n$ is
independent of previous actions and observations. The vector set
$\{(i_n,t_n)\}_{n=1}^N$ is thus a sufficient statistic. The one-step
update of $\{(i_n,t_n)\}_{n=1}^N$ is straightforward.
\end{proof}
Now we complete the RMAB formulation of the intrusion detection
problem by observing that the immediate reward $R_n(S_n)$ offered by
component $n$ can be modeled by $-c_n$ if it is currently in the
abnormal state and $0$ otherwise. Consequently, the reward
maximization is equivalent to the cost minimization. In the rest of
the paper, we use RMAB-IDS to denote this class of RMAB.

\subsection{The Optimality Equation}In this subsection, we establish the optimality equation
for RMAB-IDS. We consider the following strong average-reward
criterion under which not only the steady-state average reward but
also the transient reward starting from an arbitrary initial arm
state is maximized, leading to the maximum long-term total reward
growth rate.
\begin{eqnarray}\label{eqn:dynamic}
G+F(\{(i_n,t_n)\}_{n=1}^N)=\max_{\Ac}
\mathbb{E}_{\Ac}[\sum_{n=1}^NR_n(S_n)\\\nn+F(\{\Gamma(i_n,t_n|a_n,S_n)\}_{n=1}^N)],
\end{eqnarray}where $\Ac=\{a_n\}_{n=1}^N$ with $\sum_{n=1}^Na_n=K$
denotes the current probing actions, $G$ the maximum steady-state
average reward over the infinite horizon, $F(\cdot)$ the transient
reward starting from the initial arm states, and
$\mathbb{E}_{\Ac}[\cdot]$ the expectation operator given $\Ac$.
Solving the optimality equation~\eqref{eqn:dynamic} suffers from the
curse of dimension and has an exponential complexity for dynamic
programming. In Sec.~\ref{sec:results}, we show that for RMAB-IDS,
the linear-complexity Whittle index policy exists and can be
obtained in closed form with a near-optimal performance.

\subsection{Definition of Whittle Index Policy} The key idea of
Whittle index policy is to provide a subsidy for passivity to
measure the attractiveness of activating an arm based on its current
state. Based on the strong decomposability of Whittle index, it is
sufficient to focus on each single arm~\cite{whittle}.

\subsubsection{Single-Armed Bandit with Subsidy} Consider the
single-armed bandit for the intrusion detection problem with only
one arm/component. At each time instant, we decide whether to
activate the arm or make it passive. Assume that a subsidy for
passivity, denoted by $\lambda$, is gained whenever the arm is made
passive. We have the following optimality equations. For simplicity
of presentation, we will drop the component index from the
notations.
\begin{eqnarray}\nn
g+f(0,t)&=&\max\{\lambda-p(t)c+f(0,t+1),\\\nn&&-p(t)c+p(t)f(1,1)+(1-p(t))f(0,1)\}\\\label{eqn:dynamicSingle}
&=&\max\{\lambda+f(0,t+1),\\\nn &&
p(t)f(1,1)+(1-p(t))f(0,1)\},\\\label{eqn:dynamicSingle1}
g+f(1,t)&=&\max\{\lambda+f(1,t+1),\\\nn &&
p(t-1)f(1,1)+(1-p(t-1))f(0,1)\},
\end{eqnarray}where $g$ and $f(\cdot)$ denote, respectively, the maximum
steady-state average reward and the transient reward by playing the
single arm. The optimal policy for this single-arm problem is
essentially given by an optimal partition of the arm state space
$\bigcup_{i=0,1}\{(i,~t)\}_{t\ge1}$ into a passive set
\begin{eqnarray}\nn
\Pc(\lambda)&=&\{(i,~t):a^*(i,~t,~\lambda)=0\}\\\nn
&=&\{(i,~t):\lambda+f(i,t+1)\\\nn&&\ge
p(t-i)f(1,1)+(1-p(t-i))f(0,1)\}
\end{eqnarray}
and its complement, an active set
$\Ac(\lambda)=\{(i,~t):a^*(i,~t,~\lambda)=1\}$, where
$a^*(i,~t,~\lambda)$ denotes the optimal action at arm state
$(i,~t)$ under subsidy $\lambda$.

\subsubsection{Indexability and Whittle Index} To define Whittle index policy,
it is required that the RMAB is {\em indexable}~\cite{whittle}.
\begin{definition}\label{def:indexability}
An RMAB is {\em indexable} if for each arm, the passive set
$\Pc(\lambda)$ increases monotonically from the empty set $\phi$ to
the entire state space $\bigcup_{i=1,2}\{(i,~t)\}_{t\ge1}$ as the
subsidy $\lambda$ increases from $-\infty$ to $+\infty$. An RMAB is
{\em strictly indexable} if the states join the passive set one by
one (instead of as groups) as $\lambda$ continuously increases.
\end{definition}

Given the indexability, the Whittle index $W(i,~t)$ of an arm state
$(i,~t)$ is defined as the infimum subsidy $\lambda$ that makes the
passive action optimal at $(i,~t)$:
\begin{eqnarray}\nn
W(i,~t)&=&\inf\{\lambda:~a^*(i,~t,~\lambda)=0\}\\\nn
&=&\inf\{\lambda:~\lambda+f(i,t+1)\\\nn&&\ge
p(t-i)f(1,1)+(1-p(t-i))f(0,1)\}.
\end{eqnarray}

Whittle index essentially measures how attractive it is to activate
an arm based on subsidy $\lambda$. The minimum subsidy $\lambda$
that is needed to move an arm state from the active set to the
passive set under the optimal partition thus measures how attractive
this arm state is.

Whittle index policy is naturally given by playing the $K$ arms with
the largest Whittle indexes.

\section{Indexability and the Closed-Form Whittle Index for RMAB-IDS}\label{sec:results}
In this section, we establish the indexability of RMAB-IDS and solve
for Whittle index in closed form. Based on the indexability and
Whittle index, we study the optimal policy for RMAB-IDS under a
relaxed constraint.

\subsection{Indexability}\label{subsec:indexability}
\begin{theorem}\label{thm:indexability}
RMAB-IDS is indexable.
\end{theorem}
\begin{proof}
Consider the single-armed bandit with subsidy. Without loss of
generality, we assume that the cost $c=1$. Define stopping time
$t_i$ as the number of steps until the first activation after
observing the arm in component state $i\in\{0,1\}$. We can rewrite
the dynamic equations~\eqref{eqn:dynamicSingle}
and~\eqref{eqn:dynamicSingle1} as follows.
\begin{eqnarray}\nn
f(0)&=&\max_{t_0\ge1}\{-gt_0+\lambda(t_0-1)-
\sum_{k=1}^{t_0}p(k)~~~~\\\nn&&+p(t_0)f(1)+(1-p(t_0))f(0)\},\\\nn
f(1)&=&\max_{t_1\ge1}\{-gt_1+\lambda(t_1-1)-
\sum_{k=1}^{t_1}p(k-1)~~~~\\\nn&&+p(t_1-1)f(1)+(1-p(t_1-1))f(0)\},
\end{eqnarray}where $f(i)~(i\in\{0,1\})$ is the transient reward
starting from arm state $(i,0)$. Note that we can set $f(0)=0$ since
only $f(1)-f(0)$ is determined by the above equations. We thus have
\begin{eqnarray}\label{eqn:dynamicTime}
0&=&\max_{t_0\ge1}\{-gt_0+\lambda(t_0-1)-
\sum_{k=1}^{t_0}p(k)~~~~\\\nn&&+p(t_0)f(1)\},\\\label{eqn:dynamicTime1}
f(1)&=&\max_{t_1\ge1}\{-gt_1+\lambda(t_1-1)-
\sum_{k=1}^{t_1}p(k-1)~~~~\\\nn&&+p(t_1-1)f(1)\}.
\end{eqnarray}

To prove indexability, it is equivalent to prove that the optimal
$\{t^*_i\}_{i=0,1}$ in ~\eqref{eqn:dynamicTime}
and~\eqref{eqn:dynamicTime1} are nondecreasing with $\lambda$. For
the case that $\lambda<0$, all states are in the active set, \ie
$t^*_i=1$ for $i\in\{0,1\}$. This is because that both the time
portion of the occurrence of the abnormal component state and the
passive time are minimized by always activating the arm.

Consider the case that $\lambda\ge0$. We should always make the arm
passive if the observation of the component state in the previous
slot is $1$, since the current component state is guaranteed to be
$0$ after repair and there is no benefit to observe it again.
Consequently, $t_1^*>1$. Combined with~\eqref{eqn:dynamicTime}
and~\eqref{eqn:dynamicTime1}, we further observe that
$t_1^*=t_0^*+1$. Note that this holds not only for the optimal
stopping times $\{t_i^*\}_{i=0,1}$ but also for all stationary
policies with $t_1>1$. By considering $t^*_i$
in~\eqref{eqn:dynamicTime} and~\eqref{eqn:dynamicTime1}, we can
solve for $f(1)$ and $g$ and obtain
\begin{eqnarray}\label{eqn:g}
g=\frac{\lambda(t_0^*-1+p(t_0^*))-\sum_{k=1}^{t_0^*}p(k)}{t_0^*+p(t_0^*)}.
\end{eqnarray}

Now suppose that it is better to activate the arm at the $t$-th step
instead of any earlier step after observing component state $0$. We
have
\begin{eqnarray}\nn
\frac{\lambda(t-1+p(t))-\sum_{k=1}^{t}p(k)}{t+p(t)}\\\label{eqn:compare}
\ge\frac{\lambda(s-1+p(s))-\sum_{k=1}^{s}p(k)}{s+p(s)},~\forall~s\in\{1,\cdots,t\}.
\end{eqnarray}We can further simplify~\eqref{eqn:compare} and obtain
for all $s\in\{1,\cdots,t\}$,
\begin{eqnarray}\nn
\lambda(t-s+p(t)-p(s))\\\label{eqn:hold}
\ge\sum_{k=1}^{t}p(k)(s+p(s))-\sum_{k=1}^{s}p(k)(t+p(t)).
\end{eqnarray}Based on the monotone property of $\{p(t)\}_{t\ge0}$,
we have $t-s+p(t)-p(s)\ge0$ and~\eqref{eqn:hold} keeps true as
$\lambda~(\lambda\ge0)$ increases. Equivalently, the set of $t$ for
which~\eqref{eqn:compare} and~\eqref{eqn:hold} are true is
nondecreasing in $\lambda$. We thus conclude that
$\{t^*_i(\lambda)\}_{i=0,1}$ are nondecreasing in $\lambda$. Since
this further implies that $\Pc(\lambda)$ is nondecreasing in
$\lambda$, we proved the indexability.
\end{proof}

\subsection{The Closed-Form Whittle Index}\label{subsec:whittleindex}
Given the indexability established in Sec~\ref{subsec:indexability},
we proceed to solve for the closed-form Whittle index of RMAB-IDS.
For simplicity of presentation, we focus on the case that the bandit
is strictly indexable (see Definition~\ref{def:indexability}), \ie
there is no tie among the Whittle indexes. A simple condition in the
following is adopted to guarantee the strict indexability.

\noindent \emph{C1:} $p(t+1)-p(t)$ is strictly decreasing with $t$.

Note that C1 is always satisfied under the Markovian state model
(see Sec.~\ref{sec:model}). As shown in the following theorem, under
C1, RMAB-IDS is strictly indexable. The closed-form Whittle index
function is subsequently obtained.
\begin{theorem}\label{thm:whittleindex}
Under C1, RMAB-IDS is strictly indexable and the Whittle index
$W(\cdot)$ is given below.
\begin{eqnarray}\label{eqn:whittleindex}
W(0,t)&=&(\frac{p(t+1)(t+p(t))}{1+p(t+1)-p(t)}-\sum_{k=1}^tp(k))c,~~\\\label{eqn:whittleindex2}
W(1,t)&=&W(0,t-1),~W(0,0)\defeq0.
\end{eqnarray}
\end{theorem}
\begin{proof}
We first prove the following lemma that establishes a sufficient and
necessary condition for strict indexability and the associated
Whittle index.
\begin{lemma}\label{lemma:strict}
Define $W(0,t)$ as in~\eqref{eqn:whittleindex}. RMAB-IDS is strictly
indexable if and only if $W(0,t)$ is strictly increasing with $t$.
In this case, the Whittle index of state $(i,t)~(i\in\{0,1\})$ is
given by~\eqref{eqn:whittleindex} and~\eqref{eqn:whittleindex2}.
\end{lemma}
\begin{proof}
Without loss of generality, we assume that the cost $c=1$. We first
prove the necessity. If the bandit is strictly indexable, the states
$\{(0,t)\}_{t\ge1}$ join the passive set one by one as $\lambda$
continuously increases. From the proof of
Theorem~\ref{thm:indexability}, after observing component state $0$,
it is optimal to activate the arm at the $t$-th step under subsidy
$\lambda$ if and only if
\begin{eqnarray}\label{eqn:hold1} \lambda \ge
\frac{d(t,s)}{c(t,s)},~\forall~s<t,\\\label{eqn:hold2} \lambda \le
\frac{d(u,t)}{c(u,t)},~\forall~u>t,
\end{eqnarray}where
{\small\begin{eqnarray}\nn c(x,y)&\defeq& x-y+p(x)-p(y),\\\nn
d(x,y)&\defeq&\sum_{k=1}^{x}p(k)(y+p(y))-\sum_{k=1}^{y}p(k)(x+p(x)).
\end{eqnarray}}Consider an arbitrary $v\ge1$. If
both~\eqref{eqn:hold1} and~\eqref{eqn:hold2} hold with equality by
letting $(u,t,s)=(v+2,v+1,v)$ and $\lambda=W(0,v)$, than Whittle
indexes for states $(0,v)$ and $(0,v+1)$ would be the same. This
contradicts the strict indexability. We thus have that
$d(v+1,v)/c(v+1,v)$ is strictly increasing at $v$.

Now we prove the sufficiency. Assume that $W(0,t)$ is strictly
increasing with $t$. This implies that $W(0,t)$ is positive for all
$t$ since
\[W(0,1)=p(2)-p(1)+p^2(1)>0.\]For an arbitrary $v\ge 1$, there must exist a
subsidy $\lambda>0$ such that both~\eqref{eqn:hold1}
and~\eqref{eqn:hold2} hold with strict inequality by letting
$(u,t,s)=(v+2,v+1,v)$. So the Whittle index for state $(0,v)$ is
smaller than this $\lambda$ while the Whittle index for state
$(0,v+1)$ is larger than it. This proves the strict indexability.

Under the strict indexability, if we set the subsidy $\lambda_t$ as
the Whittle index of state $(0,t)$, then it is optimal to either
activate on $(0,t)$ or wait one more step to activate on $(0,t+1)$.
We thus have
\begin{eqnarray}
\lambda_t c(t+1,t)=d(t+1,t),
\end{eqnarray}which leads to the Whittle index of state $(0,t)$ as
given in~\eqref{eqn:whittleindex}. Recall that for any nonnegative
subsidy, the optimal activation time after observing component state
$1$ is one step later compared to that after observing component
state $0$. we arrive at $W(1,t)=W(0,t-1)$ for $t\ge2$. Based on the
proof of Theorem~\ref{thm:indexability}, it is not hard to see that
$W(1,1)=0$. We thus proved the lemma.
\end{proof}
\vspace{.5em}

Based on Lemma~\ref{lemma:strict}, we only need to prove that C1
implies the strict monotonicity increasing property of $W(0,t)$.
Equivalent, for any $t\ge1$, we need to prove
\begin{eqnarray}\label{eqn:ineq}
\frac{d(t+2,t+1)}{c(t+2,t+1)}>\frac{d(t+1,t)}{c(t+1,t)}.
\end{eqnarray}Define $\delta(t)\defeq p(t+1)-p(t)$ which is positive under C1. By
simplifying~\eqref{eqn:ineq}, it is equivalent to prove
\begin{eqnarray}\nn
&p(t+1)t\delta(t)+p^2(t+1)\delta(t)
+\delta(t)\delta(t+1)\\\nn&+p(t+1)t+p^2(t+1) +\delta(t+1)(t+1)\\\nn
&>p(t)t\delta(t+1)+p^2(t)\delta(t+1)\\\label{eqn:ineq1}&+\delta(t)\delta(t+1)p(t)
+p(t)t+p^2(t)+\delta(t)p(t)+\delta(t)t.~~~
\end{eqnarray}Since $p(t)$ is increasing and $\delta(t)$ is strictly
decreasing with $t$ (under C1), we have
\begin{eqnarray}\nn
&p(t+1)t\delta(t)+p^2(t+1)\delta(t)
+\delta(t)\delta(t+1)\\\label{eqn:ineq2}&
>p(t)t\delta(t+1)+p^2(t)\delta(t+1)+\delta(t)\delta(t+1)p(t).~~
\end{eqnarray}To prove~\eqref{eqn:ineq1}, it is sufficient to prove
\begin{eqnarray}\nn
&p(t+1)t+p^2(t+1)+\delta(t+1)(t+1)\\\label{eqn:ineq3}&>p(t)t+p^2(t)+\delta(t)p(t)+\delta(t)t.
\end{eqnarray}After some simplifications of~\eqref{eqn:ineq3}, we need to prove
\begin{eqnarray}\label{eqn:ineq4}
\delta(t)p(t+1)+\delta(t+1)(t+1)>0,
\end{eqnarray}which is always true under C1. We thus proved
Theorem~\ref{thm:whittleindex}.
\end{proof}

The near-optimal performance of Whittle index policy is observed
through numerical examples (see Sec.~\ref{sec:num}). In
Sec.~\ref{sec:optwhittle}, we show that when all components are
homogeneous, Whittle index policy is equivalent to the myopic policy
and achieves the optimal performance.

\subsection{The Optimal Policy under a Relaxed Constraint}
In this subsection, we consider the scenario with a relaxed resource
constraint, where we only require the {\em average} number of
activated arms to be no more than $K$. This scenario often arises in
systems where the resource constraint is more strict on the average
value rather than the peak value, \eg the energy-saving systems.
Under the relaxed constraint, the indexability and the Whittle index
leads to a simple optimal policy for RMAB-IDS.

As explained by Whittle in~\cite{whittle}, the subsidy $\lambda$ for
passivity is essentially the Lagrangian multiplier for the general
RMAB with the following relaxed constraint
\begin{eqnarray}\label{eqn:c1}
\mathbb{E}_{\pi}\left[\lim_{T\rightarrow\infty}\frac{1}{T}\sum_{t=1}^TK(t)\right]\le
K,
\end{eqnarray}where $K(t)$ is the number of activated arms at time
$t$. Specifically, the subsidy $\lambda$ controls the expected time
portion, \ie the stead-state probability $\pi_n(\lambda)$, that arm
$n~(1\le n\le N)$ is made active under the corresponding single-arm
optimal policy. For RMAB-IDS, under the optimal subsidy $\lambda^*$,
we have
\begin{eqnarray}\label{eqn:c2}
\mathbb{E}_{\pi^*}\left[\lim_{T\rightarrow\infty}\frac{1}{T}\sum_{t=1}^TK(t)\right]=
\sum_{n=1}^N\pi_n(\lambda^*)=K
\end{eqnarray}and~\eqref{eqn:c1} is
satisfied with equality.

Given the optimal subsidy $\lambda^*$, the optimal policy under the
relaxed constraint is simply given by the composition of $N$
independent single-arm optimal policies (applied on the $N$ arm
respectively) under the common subsidy $\lambda^*$. Specifically, at
each time, if the Whittle index of an arm is larger than $\lambda^*$
then we activate the arm; otherwise we make the arm passive. Note
that if the Whittle index of an arm is equal to $\lambda^*$,
randomizing between the active and passive actions would be
necessary to satisfy~\eqref{eqn:c2} as detailed in~\cite{Weber}.
Given the closed-form Whittle index established in
Theorem~\ref{thm:whittleindex}, it remains to solve for the optimal
subsidy $\lambda^*$. Note that based on the Lagrangian multiplier
theorem~\cite{whittle}, we have
\begin{eqnarray}\label{eqn:optimalsubsidy}
\lambda^*=\arg\min_{\lambda}\{\sum_{n=1}^Ng_n(\lambda)-(N-K)\lambda\},
\end{eqnarray}where $g_n(\lambda)$ is the maximum average reward of arm
$n$ under the single-arm policy for subsidy $\lambda$ and is convex
in $\lambda$. From the closed-form Whittle index, it is not hard to
solve for the optimal stopping times $\{t_i^*(\lambda)\}_{i=0,1}$
(see~\eqref{eqn:g}) and the maximum average reward $g(\lambda)$ for
each $\lambda$. We can then obtain the optimal $\lambda^*$
from~\eqref{eqn:optimalsubsidy} by any classic algorithm for finding
the minimum of a convex function.

\section{Optimality in Homogeneous Networks}\label{sec:optwhittle}
In this section, we study the performance of Whittle index policy in
homogeneous networks, \ie all components have the same parameters:
the probability sequence $\{p(t)\}_{t\ge0}$ and the per-unit cost
$c$ for being abnormal.

We first establish the equivalence of Whittle index policy with the
myopic policy for homogeneous components. In general, the myopic
policy chooses the $K$ components to solely minimize the expected
cost in the next slot. It is not hard to show that for homogeneous
components, the myopic policy is reduced to choosing the $K$
components with the largest probabilities of being in the abnormal
state. The myopic action $\hat{\Ac}(\cdot)$ as a function of the
current states of all arms is thus given below.
{\small\begin{eqnarray}\nn
\hat{\Ac}(\{i_n,t_n\}_{n=1}^N)&=&\arg\max_{\Ac}\{\sum_{n:a_n=1}
\Pr(S_n=1|(i_n,t_n))\}\\\nn&=&\arg\max_{\Ac}\{\sum_{n:a_n=1}(p(t_n)(1-i_n)\\\label{eqn:myopic1}
&&+p(t_n-1)i_n)\}.
\end{eqnarray}}
\begin{lemma}\label{lemma:equivalence}
For homogeneous components, Whittle index policy is equivalent to
the myopic policy and has the following simple structure: initialize
a queue in which components are ordered according to the descending
order of their initial probabilities of being in the abnormal state.
Each time we probe the $K$ components at the head of the queue. In
the next slot, these $K$ components will be moved to the bottom of
the queue while keeping those observed in state $1$ a higher
position than those observed in state $0$.
\end{lemma}
\begin{proof}
Based on the proof of Theorem~\ref{thm:indexability}, the Whittle
index $W(i,t)$ of an arm is monotonically increasing with $t$ for
fixed $i\in\{0,1\}$ and $W(1,t)=W(0,t-1)$ with $W(0,0)=0$. Based on
the monotonic increasing property of $\{p(t)\}_{t\ge0}$, it is not
hard to see that the Whittle index $W(i,t)$ is monotonically
increasing with $\Pr(S=1|(i,t))$. Whittle index policy is thus
equivalent to the myopic policy for homogeneous arms.

From the equivalence of Whittle index policy with the myopic policy,
its structure is straightforward since based on the current
observations, all components observed in state $1$ will have zero
probability of being abnormal and those observed in state $0$ will
have the second smallest probability $p(1)$ of being abnormal, while
those unobserved arms will have the same rank in the probability of
being abnormal in the next slot based on the monotonicity of
$\{p(t)\}_{t\ge0}$.
\end{proof}

From Lemma~\ref{lemma:equivalence}, Whittle index policy can be
implemented without knowing the system parameters $\{p(t)\}_{t\ge0}$
and $c$. Furthermore, Whittle index policy is optimal, as given in
the following theorem.

\begin{theorem}\label{thm:opt}
For homogeneous components, Whittle index policy minimizes the
expected total cost over a finite time horizon of an arbitrary
length $T~(T\ge1)$. It is thus also optimal under the strong
average-reward criterion over the infinite time horizon.
\end{theorem}
\begin{proof}
We prove the theorem based on a backward induction on the time
horizon. Any policy, including Whittle index policy, is optimal at
the last time instant $t=T$ since the current action affects only
the future cost but not the immediate cost. Now assume that Whittle
index policy is optimal at time instants $t+1,t+2,\cdots,T$. We need
to prove that it is optimal at time $t$. Without loss of generality,
we set $c=1$. Let $\Omega(t)=(\omega_1,\omega_2,\cdots,\omega_N)$
with $\omega_n\in\{p(t)\}_{t\ge0}$ denote an unordered set
consisting of probabilities that the $N$ components are in state $1$
at time $t$. Define the value function $V_t(\Omega(t))$ of Whittle
index policy as the expected total cost from time $t$ up to $T$.
Consider a policy that activate the $K$ components with
probabilities $(\omega_1,\omega_2,\cdots,\omega_K)$ of being in
state $1$ at time $t$ and follows Whittle index policy in the future
time instants up to time $T$. The value function
$\hat{V_t}(\omega_1,\omega_2,\cdots,\omega_N)$, \ie the expected
total cost from time $t$ up to $T$, of this policy is given by
\begin{eqnarray}\nn
\hat{V_t}(\omega_1,\omega_2,\cdots,\omega_N)=\sum_{k=1}^N\omega_k
\\\nn+\mathbb{E}[V_{t+1}(\underbrace{0,\cdots,0}_{k_1~\mbox{times}},
\underbrace{p(1),\cdots,p(1)}_{k_0~\mbox{times}},\cdots,\tau(\omega_N))],
\end{eqnarray}where the expectation is taken over the random variables $\{k_i\}_{i=0,1}~
(k_1+k_0=K)$ that denote respectively the number of components
observed in state $1$ and state $0$, and $\tau(\cdot)$ denote the
one-step update of the abnormal probability for unobserved
components based on $\{p(t)\}_{t\ge0}$. Note that if $\omega_1\ge
\omega_2\ge\cdots \omega_N$, then $\hat{V_t}=V_t$.

To prove that Whittle index policy, \ie the myopic policy, is
optimal at time $t$, it is sufficient to prove that for any $y\ge
x,~x,y\in\{p(t)\}_{t\ge0}$,
\begin{eqnarray}\nn
\hat{V_t}(\omega_1,\cdots,y,\cdots,x,\cdots,\omega_N)\\\label{eqn:opt}\le
\hat{V_t}(\omega_1,\cdots,x,\cdots,y,\cdots,\omega_N).
\end{eqnarray}This means that a component with higher probability of
being in state $1$ should be given a higher priority. To
show~\eqref{eqn:opt}, we first present the following lemma that
establishes the monotonicity of the value function of Whittle index
policy.
\begin{lemma}\label{lemma:monotone}
The value function $V_t(\omega_1,\omega_2,\cdots,\omega_N)$ of
Whittle index policy is an increasing function at each entry
$\omega_n~(n\in\{1,2,\cdots,N\})$.
\end{lemma}
\begin{proof}
Without loss of generality, we assume that all probabilities within
$V_t(\cdot)$ are in a descending order. The proof is based on a
backward induction on time $t$. If $t=T$, the claim is clearly true.
Assume that the lemma holds for $s=t+1,t+2,\cdots,T$. Consider time
$t$. We need to show
\begin{eqnarray}\label{eqn:monotone}
V_t(\overrightarrow{\omega_1},y,\overrightarrow{\omega_2})\ge
V_t(\overrightarrow{\omega_1},x,\overrightarrow{\omega_2}),~\forall~y\ge
x,~x,y\in\{p(t)\}_{t\ge0},
\end{eqnarray}where
$\overrightarrow{\omega_1},\overrightarrow{\omega_2}$ are arbitrary
(possibly empty) probability vectors with
$|\overrightarrow{\omega_1}|+|\overrightarrow{\omega_2}|=N-1$.

Define $t_1\ge1$ as the first stopping time that the component
denoted by $y/x$ in~\eqref{eqn:monotone} is probed under Whittle
index policy. Based on the structure of Whittle index policy, $t_1$
is deterministic. We have {\begin{eqnarray}\label{eqn:v1}\small
V_t(\vec{\omega}_1,y,\vec{\omega}_2)&=&
A(\vec{\omega_1},\vec{\omega_2}) +\sum_{k=1}^{t_1}\tau^{k-1}(y)\\\nn
&&+\mathbb{E}[\tau^{t_1-1}(y)V_{t+t_1}(\vec{\omega}'_1,0,\vec{\omega}'_2)\\\nn
&&+(1-\tau^{t_1-1}(y))V_{t+t_1}(\vec{\omega}'_1,p(1),\vec{\omega}'_2)],\\\label{eqn:v2}
V_t(\vec{\omega}_1,x,\vec{\omega}_2)&=&
A(\vec{\omega}_1,\vec{\omega}_2) +\sum_{k=1}^{t_1}\tau^{k-1}(x)\\\nn
&&+\mathbb{E}[\tau^{t_1-1}(x)V_{t+t_1}(\vec{\omega}'_1,0,\vec{\omega}'_2)\\\nn
&&+(1-\tau^{t_1-1}(x))V_{t+t_1}(\vec{\omega}'_1,p(1),\vec{\omega}'_2)],
\end{eqnarray}}where
$A(\vec{\omega_1},\vec{\omega_2})$ is the expected total cost up to
$t_1$ determined by components other than that denoted by $y/x$,
vectors $\vec{\omega}'_1,\vec{\omega}'_2)$ are stochastically
determined by $\vec{\omega}_1,\vec{\omega}_2$ based on the
observations between time $t$ and $t+t_1-1$, and $\tau^{k}(\cdot)$
denotes the $k$-th iteration of operator $\tau(\cdot)$. We point out
that based on the structure of Whittle index policy, the total cost
$A(\vec{\omega_1},\vec{\omega_2})$ does not depend on the state of
the component denoted by $y/x$. From~\eqref{eqn:v1}
and~\eqref{eqn:v2}, we have that~\eqref{eqn:monotone} holds if
{\begin{eqnarray}\nn\small
\sum_{k=1}^{t_1-1}(\tau^{k-1}(y)-\tau^{k-1}(x))
+(\tau^{t_1-1}(y)-\tau^{t_1-1}(x))
\mathbb{E}[\\\label{eqn:v3}1+V_{t+t_1}(\vec{\omega}'_1,0,\vec{\omega}'_2)
-V_{t+t_1}(\vec{\omega}'_1,p(1),\vec{\omega}'_2)]\ge0.
\end{eqnarray}}From the monotonic increasing property of $\{p(t)\}_{t\ge0}$, we
have \[\tau^{k}(y)-\tau^{k}(x)\ge0,~\forall~y\ge
x,~x,y\in\{p(t)\}_{t\ge0},~k\ge0.\] To show~\eqref{eqn:v3}, it is
sufficient to show
\begin{eqnarray}\label{eqn:v4}
\mathbb{E}[1+V_{t+t_1}(\vec{\omega}'_1,0,\vec{\omega}'_2)
-V_{t+t_1}(\vec{\omega}'_1,p(1),\vec{\omega}'_2)]\ge0.
\end{eqnarray}Starting from time $t+t_1$, define $t_2$ as the first stopping time that the
component denoted by $0/p(1)$ is probed under Whittle index policy.
Between time $t+t_1$ to $t+t_1+t_2$, the difference in the expected
total cost incurred by this component when its abnormal
probabilities are respectively given by $0$ and $p(1)$ is equal to
$p(t_2)$. This is because that the update of the abnormal
probability when staring from $0$ is one step lagged of that from
$p(1)$. Again, based on the structure of Whittle index policy, the
expected total cost incurred by other components is independent of
the state of this component. By expanding the value function
in~\eqref{eqn:v4} at time $t+t_1+t_2$ and after some
simplifications, it is equivalent to show
\begin{eqnarray}\nn\label{eqn:v5}
\mathbb{E}[1-p(t_2)+(p(t_2-1)-p(t_2))
\mathbb{E}[\\\label{eqn:v5}V_{t+t_1+t_2}(\vec{\omega}''_1,0,\vec{\omega}''_2)
-V_{t+t_1+t_2}(\vec{\omega}''_1,p(1),\vec{\omega}''_2)]\ge0,
\end{eqnarray}where vectors $\vec{\omega}'_1,\vec{\omega}'_2)$ are stochastically
determined by $\vec{\omega}_1,\vec{\omega}_2$ based on observations
between time $t+t_1$ and $t+t_1+t_2-1$. By induction, for any
$\vec{\omega}''_1,\vec{\omega}_2''$,
\begin{eqnarray}\nn
V_{t+t_1+t_2}(\vec{\omega}''_1,0,\vec{\omega}''_2)
-V_{t+t_1+t_2}(\vec{\omega}''_1,p(1),\vec{\omega}''_2)\le 0.
\end{eqnarray}It is thus not hard to see that~\eqref{eqn:v5} holds.
Note that for the realizations of $t_1$ and $t_2$ such that
$t+t_1>T$ and/or $t+t_1+t_2>T$, the monotonicity of the conditional
value function is straightforward to prove. We thus proved the
lemma.
\end{proof}
Now we are ready to prove~\eqref{eqn:opt}. If the positions of $y$
and $x$ are both in top $K$ or both after top $K$, then the
inequality holds with equality. Consider the case that $y$ is in top
$K$ but $x$ not. We have for any probability vectors
$\{\vec{\omega}_i\}_{i=1,2,3}$,
\begin{eqnarray}\nn
&&~~~~~~~~~~~~~~~~~~~~~~\hat{V_t}(\vec{\omega}_1,y,\vec{\omega}_2,x,\vec{\omega}_3)\\\nn
&=&\mathbb{E}[yV_{t+1}(\vec{\omega}'_1,\vec{\omega}'_2,\tau(x),\vec{\omega}'_3,0)\\\nn
&&+(1-y)V_{t+1}(\vec{\omega}'_1,\vec{\omega}'_2,\tau(x),\vec{\omega}'_3,p(1))]\\\nn
&\le&\mathbb{E}[yV_{t+1}(\vec{\omega}'_1,\vec{\omega}'_2,\tau(y),\vec{\omega}'_3,0)\\\nn
&&+(1-y)V_{t+1}(\vec{\omega}'_1,\vec{\omega}'_2,\tau(y),\vec{\omega}'_3,p(1))]\\\nn
&\le&\mathbb{E}[xV_{t+1}(\vec{\omega}'_1,\vec{\omega}'_2,\tau(y),\vec{\omega}'_3,0)\\\nn
&&+(1-x)V_{t+1}(\vec{\omega}'_1,\vec{\omega}'_2,\tau(y),\vec{\omega}'_3,p(1))]\\\nn
&=&\hat{V_t}(\vec{\omega}_1,x,\vec{\omega}_2,y,\vec{\omega}_3),
\end{eqnarray}where
$\vec{\omega}'_1,\vec{\omega}'_2,\vec{\omega}'_3$ are stochastically
determined by $\vec{\omega}_1,\vec{\omega}_2,\vec{\omega}_3$ based
on the observation at time $t$, and the two inequalities are due to
Lemma~\ref{lemma:monotone}. We thus proved the optimality of Whittle
index policy over a finite horizon of an arbitrary length $T$. By
contradiction, if Whittle index policy is not optimal under the
strong average-reward criterion, there must exist a $T_0$ such that
Whittle index policy performs worse than the optimal policy over the
horizon of length $T_0$. Consequently, Whittle index policy is also
optimal under the strong average-reward criterion over the infinite
time horizon.
\end{proof}

\section{Numerical Examples}\label{sec:num}
In this section, we present some numerical examples and evaluate the
performance of Whittle index policy for nonhomogeneous components.

In Fig.~\ref{fig:WI}, we illustrate the Whittle index as a function
of the arm state. The monotonicity and concavity of the Whittle
index are observed. In Fig.~\ref{fig:perform2}, we compare the
performance of Whittle index policy versus the optimal policy. Due
to the complexity of the dynamic programming problem given
in~\eqref{eqn:dynamic}, we only computed the optimal cost over a
short time horizon. Note that the cost under the non-stationary
optimal policy over a finite time horizon is a lower bound on that
achieved by the stationary optimal policy over the infinite time
horizon. We observe that Whittle index policy achieves a
near-optimal performance.

In Fig.~\ref{fig:perform4}, we compare Whittle index policy with the
myopic policy over a long time horizon. We observe that for
inhomogeneous components, Whittle index policy outperforms the
myopic policy, and the performance improvement becomes significant
as time goes.

Numerical results similar to the above have been observed through
extensive examples with randomly generated system parameters.

\begin{figure}
\centerline{
\begin{psfrags}
\scalefig{.4}\epsfbox{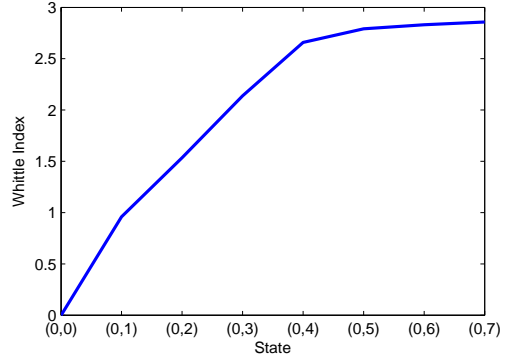}
\end{psfrags}}
\caption{The Whittle index ($\{p(t)\}_{0\le t\le8}=$[0,0.5,0.7,0.85,
0.95,0.97,.975,.978,.98], $c=1$).} \label{fig:WI}
\end{figure}

\begin{figure}
\centerline{
\begin{psfrags}
\scalefig{.39}\epsfbox{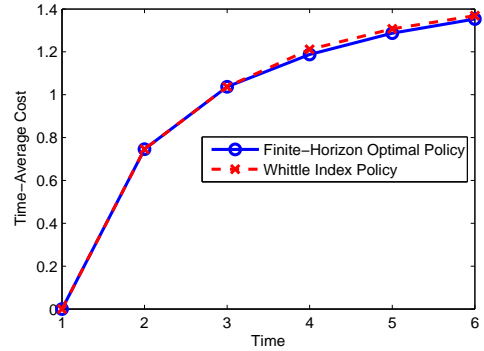}
\end{psfrags}}
\caption{The near-optimality of Whittle index policy
($K=1$,~$N=4$,~$\{p_n(t)\}_{n=1,2,\cdots,4,0\le t\le6}=$
$[0,.5,.7,.85,.95,.97,.975;0,.3,.4,.48,.54,.57,.59;$
$0,.36,.46,.5,.53,.55,.56;0,.6,.78,.9,.96,.98,.99]$,
$\{c_n\}_{n=1,2,\cdots,4}=[.8,1,1.2,.9]$, all components start from
the healthy state).} \label{fig:perform2}
\end{figure}

\begin{figure}
\centerline{
\begin{psfrags}
\scalefig{.39}\epsfbox{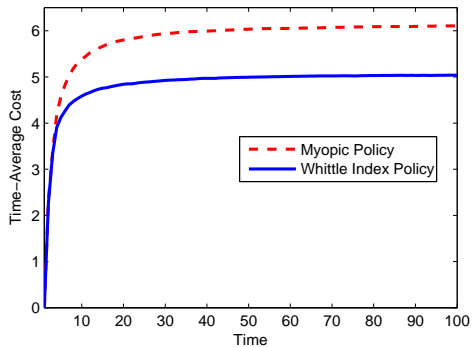}
\end{psfrags}}
\caption{The performance of Whittle index policy versus the myopic
policy ($K=2$,~$N=8$,~Markovian state model,
$\{q_n\}_{n=1,2,\cdots,8}=$ $[.2,.3,.3,.5,.6,.7,.7,.8],$~
$\{c_n\}_{n=1,2,\cdots,8}=[2.5,2,1.8,1.5,1.2,1,.6,.5]$, all
components start from the healthy state).} \label{fig:perform4}
\end{figure}

\begin{figure}
\centerline{
\begin{psfrags}
\psfrag{a}[c]{\small Random Customer
Arrivals}\psfrag{b1}[c]{\footnotesize~Buffer $1$}
\psfrag{b2}[c]{\footnotesize~Buffer
$2$}\psfrag{bn}[c]{\footnotesize~~Buffer $N$}
\psfrag{s1}[c]{\footnotesize~~Server $1$}
\psfrag{s2}[c]{\footnotesize~~Server
$2$}\psfrag{sk}[c]{\footnotesize~~~Server $K$}
\scalefig{.39}\epsfbox{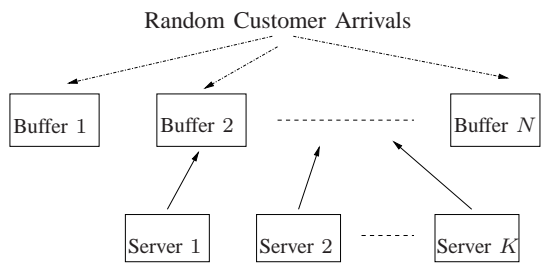}
\end{psfrags}}
\caption{The equivalent queuing model of RMAB-IDS.} \label{fig:q}
\end{figure}

\section{Applications to Queuing Networks}\label{sec:queue}
Another application of the RMAB model considered in this paper is on
holding cost minimization in queueing networks. Consider a queuing
network where customers randomly arrive at $K$ servers. As shown in
Fig.~\ref{fig:q}, all servers share a set of $N$ finite-size buffers
(for $N$ different classes of customers) that are either empty or
full based on the batch arrivals. We assume that new customers of a
class do not arrive if the corresponding server is full. At each
time, each server chooses one buffer to serve and clear its packets.
The objective is to minimize the holding cost (\eg delay) of the
customers. By likening a customer arrival to an attack, it is not
hard to see that the problem can be modeled as the RMAB at hand
under certain conditions, \eg when the arrival process of each class
is i.i.d. or Markovian over time (given the buffer is empty). Such a
queuing network often arises in backorder control systems and
peer-to-peer communication networks. For example, in a backorder
control system, random orders for $N$ commodities arrive at a seller
and the seller needs to decide which $K$ commodities to check and
process the corresponding orders at a given time. For each commodity
and at each time, a backorder incurs a cost depending on the level
of urgency and/or value of the order. In a peer-to-peer
communication network, there are $N$ communication links where each
link has a pair of nodes exchanging messages. At each time, only $K$
links can be turned on for communications and the cost can be
modeled as the delay of each message. A potential future direction
is to study the case in which the buffer can be partially full and
new arrivals come regardless of the state of the buffer. The joint
minimization of the holding cost and the customer loss cost can be
considered. Such scenario is essentially a generalized version of
the RMAB with stochastically time-varying instantaneous cost
$c_n(t)$. It is also interesting to extend the RMAB to partial reset
models for handling more general customer arrival processes.

\section{Conclusion}\label{sec:con}
In this paper, we studied the intrusion detection problem in large
cyber networks under general attack processes. By adopting a reset
model of the network dynamics, we formulated the problem as a class
of RMAB under a strong average-reward criterion. We showed that this
class of RMAB is indexable and Whittle index can be solved in
closed-form. This result leads to a low-complexity implementation of
Whittle index policy that achieves a near-optimal performance. We
further showed that for homogeneous components, Whittle index policy
can be implemented without knowing the system parameters and is
optimal over both finite and infinite time horizons.

\end{document}